\documentclass[11pt, a4paper, leqno]{article}

\usepackage[margin=0.95in]{geometry} 
\usepackage{amsfonts, amscd, amssymb, mathtools, mathrsfs, dsfont, bbm, bbding} 
\usepackage[amsmath, amsthm, thmmarks]{ntheorem} 
\usepackage{graphicx,float} 
\usepackage{indentfirst}
\usepackage{lmodern}
\usepackage[T1]{fontenc} 
\usepackage{enumerate, listings, verbatim, paralist}
\usepackage{setspace, xspace}
\usepackage[colorlinks=true, citecolor=blue]{hyperref}
\usepackage{extarrows}
\usepackage{pdfpages}

\usepackage[square, numbers]{natbib}
\usepackage{mathtools}
\mathtoolsset{showonlyrefs=true}

\newtheorem{thm}{Theorem} [section]
\newtheorem{prop}[thm]{Proposition}
\newtheorem{coro}[thm]{Corollary}
\newtheorem{lemma}[thm]{Lemma}
\newtheorem{defn}{Definition}[section]
\newtheorem{assmp}{Assumption}[section]
\newtheorem{remark}{Remark}[section]

\newtheorem{prob}{Problem}[section]

\numberwithin{equation}{section} 

\renewcommand{\geq}{\geqslant}
\renewcommand{\leq}{\leqslant}
\renewcommand{\ge}{\geqslant}
\renewcommand{\le}{\leqslant}

\newcommand{\citeassmp}[1]{Assumption \ref{#1}}

\newcommand{\opfont}{\mathbb}

\newcommand{\BE}[2][]{\ensuremath{\operatorname{\opfont{E}}^{#1}\!\left[#2\right]}}

\newcommand{\BP}[2][]{\ensuremath{\operatorname{\opfont{P}}^{#1}\!\left(#2\right)}}

\newcommand{\R}{\ensuremath{\operatorname{\mathbb{R}}}}

\newcommand{\dd}{\ensuremath{\operatorname{d}\! }}
\newcommand{\dt}{\ensuremath{\operatorname{d}\! t}}
\newcommand{\ds}{\ensuremath{\operatorname{d}\! s}}

\newcommand{\dx}{\ensuremath{\operatorname{d}\! x}}

\newcommand{\dz}{\ensuremath{\operatorname{d}\! z}}
\newcommand{\ddp}{\ensuremath{\operatorname{d}\! p}}

\newcommand{\setw}{\mathscr{W}}

\newcommand{\dct}{\textrm{the dominated convergence theorem}\xspace}
\newcommand{\mct}{\textrm{the monotone convergence theorem}\xspace}

\newcommand{\nn}{\nonumber}

\newcommand{\esup}{\ensuremath{\mathrm{ess\:sup\:}}}

\newcommand{\barh}{\overline{H}}

\newcommand{\ep}{\varepsilon}

\newcommand{\opl}{\Phi}
\newcommand{\opll}{\Psi}

\newcommand{\ms}{m_0}

\newcommand{\wg}{\mathbf{g}}
\newcommand{\pw}{\hat{F}}

\newcommand{\Ac}{\mathcal{A}}
\newcommand{\ee}{\mathrm{e}}

\begin{document}

\title{ 
Optimal Moral-hazard-free Reinsurance\\ under Extended Distortion Premium Principles 
}
\author{Zhuo Jin\thanks{Department of Actuarial Studies and Business Analytics, Macquarie University, Australia. Email: zhuo.jin@mq.edu.au.}
\and
Zuo Quan Xu\thanks{Department of Applied Mathematics, The Hong Kong Polytechnic University, Hong Kong, China. Email: maxu@polyu.edu.hk. 
}
\and
Bin Zou\thanks{Department of Mathematics, University of Connecticut, USA. Email: bin.zou@uconn.edu.}
}

\date{This Version: February 27, 2023}
\maketitle

\vspace{-4ex}
\begin{abstract} 
We study an optimal reinsurance problem under a diffusion risk model for an insurer who aims to minimize the probability of lifetime ruin. 
To rule out moral hazard issues, we only consider moral-hazard-free reinsurance contracts by imposing the incentive compatibility constraint on indemnity functions.
The reinsurance premium is calculated under an extended distortion premium principle, in which the distortion function is not necessarily concave. 
We first show that an optimal reinsurance contract always exists and then derive two sufficient and necessary conditions to characterize it. 
Due to the presence of the incentive compatibility constraint and the nonconcavity of the distortion, the optimal contract is obtained as a solution to a double obstacle problem. 
At last, we apply the general result to study three examples and obtain the optimal contract in (semi)closed form. 
\end{abstract}

\section{Introduction} 
\label{sec:intro}

Reinsurance is a type of insurance by which an insurer transfers some of its risk to a reinsurer. Reinsurance plays a crucial role in the operations of insurers, not only for mitigating risk exposure but also for stabilizing profit and other purposes (see \citet{albrecher2017reinsurance}).
In a reinsurance contract, the insurer receives a payment from the reinsurer in case of a covered loss event, specified by an indemnity function $I: z \ge 0 \mapsto 0 \le I(z) \le z$, 
and in exchange pays a premium to the reinsurer 
under some premium principle, which is a positive functional of an indemnity $I$. 
In this paper, to rule out moral hazard in claim reporting, we impose a condition on $I$, that is, $0 \le I(z) - I(z') \le z - z'$ for all $z \ge z' \ge 0$, which is commonly referred to as the \emph{incentive compatibility} (IC) constraint (see \citet{huberman1983optimal}) or \emph{no-sabotage} condition (see \citet{carlier2003pareto}). 
A reinsurance contract is called \emph{moral-hazard-free} if its indemnity function $I$ satisfies the IC constraint. 
On the pricing side, the reinsurer applies an extended distortion premium principle to calculate the premium rate. Under such a reinsurance setup, the insurer seeks 
an optimal moral-hazard-free reinsurance contract to minimize the probability of lifetime ruin. 

The study of optimal (re)insurance, or its more general version -- risk sharing, is an essential topic in economics, finance, and insurance, dating back to the 
seminal works of \citet{arrow1963uncertainty} and \citet{borch1962equilibrium} (see \citet{albrecher2017reinsurance} and \citet{cai2020optimal} for an overview).
There are several key ingredients in the modeling of optimal reinsurance problems, including the insurer's risk process, admissible (feasible) indemnity functions, reinsurance premium principles, and optimization criteria. The combination of the IC constraint on indemnities and the extended distortion premium principle leads to a challenging optimization problem and distinguishes this paper from the existing literature. 

Given a loss of size $z \ge 0$, $I(z)$ is the amount of loss ceded to the reinsurer, and $H(z) = z - I(z)$ is the amount of loss retained by the insurer. As pointed out in \citet{xu2019optimal} and \citet{boonen2023bowley}, if $H(z)$ has a derivative greater than 1 on some region, 
the insurer is incentivized to exaggerate the loss to the reinsurer; similarly, if $I(z)$ is strictly decreasing on some region, the insurer has an incentive to underreport the loss. Both scenarios are examples of ex post moral hazard. Note that when the IC constraint is imposed, both $I(z)$ and $H(z)$ are increasing functions,\footnote{In this paper, ``increasing'' means non-decreasing, and similarly ``decreasing'' means non-increasing.} with derivatives almost everywhere (a.e.) between 0 and 1. Consequently, the insurer has no incentive to overreport or underreport the loss, and such reinsurance contracts are free of moral hazard. 
The majority of the literature on optimal (re)insurance does not explicitly require the IC constraint on indemnities, but may end with an optimal contract that automatically satisfies this condition. Such a desirable result is mostly due to specific problem formulations, e.g., the expected value premium principle and utility maximization criterion in \citet{arrow1963uncertainty}, and thus should not be taken for granted that it will always hold. 
Indeed, \citet{bernard2015optimal} do not impose the IC constraint and obtain an optimal contract that could lead to ex post moral hazard; see \citet{gajek2004reinsurance} and \citet{bernard2009optimal} for additional examples. 
Therefore, we require the IC constraint when defining admissible indemnity functions (see condition (3) in Definition \ref{def:adm_indem}).

Next, we turn our attention to premium principles. The most common choice is the expected value principle, $c(I) = (1 + \xi) \mathbb{E}[I(Z)]$, in which $Z$ denotes the loss variable, and $\xi \ge 0$ is the loading factor. This principle only considers the first moment and ignores heavy tails; furthermore, it cannot incorporate the subjective beliefs of the reinsurer on the loss distribution into pricing. One solution to these drawbacks is to consider the distortion premium principle proposed by \citet{wang1996premium}, $c(I) = \int_0^\infty \wg ( \BP{I(Z) >z} )\dz$, in which $\wg$ is an increasing and concave function, with $\wg(0) = 0$ and $\wg(1) = 1$, and is called the \emph{distortion} function. 
For instance, if $\wg$ is given by $\wg(p) = p^{1/\rho}$, $\rho \ge 1$ (which recovers the proportional hazard transform principle), the parameter $\rho$ captures the inflation on small probabilities, and the reinsurer's subjective beliefs on the loss can be modeled by $\rho$. 
Here, we further generalize Wang's distortion premium principle by removing concavity assumption on $\wg$ and by incorporating a multiplier $(1 + \theta_0)$, in which $\theta_0 \ge 0$ can be seen as a loading factor.
We call such a principle an \emph{extended} distortion premium principle, which is general enough to include most existing premium principles as special cases (see Remark \ref{rem:premium}). 
A key motivation for considering nonconcave distortions comes from behavioral economics and finance; two classical examples are Quiggin's rank-dependent utility theory (\citet{quiggin1982theory}) and Tversky and Kahneman's cumulative prospect theory (\citet{tversky1992advances}). A central component in both theories is an \emph{inverse-S shaped} probability distortion function, which is concave from 0 up to a turning point and is convex afterwards; this is justified by the empirical evidence that people often overweight small probabilities and underweight large probabilities. 
We highlight that allowing nonconcave distortion functions significantly complicates the study.

With the two key modeling features discussed in detail, we are now ready to introduce the full optimal reinsurance problem investigated in this paper. We consider a representative insurer who purchases moral-hazard-free reinsurance priced under the extended distortion premium principle, and model the insurer's controlled surplus by a diffusion process (see \citet{grandell1991aspects}). The insurer seeks an optimal reinsurance contract that minimizes the probability of lifetime ruin, which is defined as the first hitting time of surplus reaching zero or below. 
Minimizing ruin probability is a popular criterion when studying risk management related problems, such as optimal reinsurance problems for insurers (see \cite{schmidli2001optimal, bayraktar2007minimizing, bayraktar2015stochastic, promislow2005minimizing, liang2020minimizing} for a short list). Such a conservative objective seems particularly appropriate for insurers, since their solvency is critical to societal stability. 

We proceed to summarize the main results of this paper. We first apply the standard dynamic programming principle to analyze the insurer's optimal reinsurance problem and obtain the associated Hamilton-Jacobi-Bellman (HJB) equation in Theorem \ref{thm:HJB}. 
By using the dual distribution $\hat{F}$ of the loss distribution $F$ (defined in \eqref{eq:F_hat}), we transform the original problem of the indemnity function $I$ into an equivalent minimization problem of the retention function $H$ (see Lemma \ref{lem:equi}), whose value function $v$ is defined in \eqref{eq:new_prob}. 
Next, we prove several key analytical properties of $v$ in Lemma \ref{opi1} and formally show that the equivalent problem of $v$ always admits an optimal solution $H^*$ in Theorem \ref{thm:existence}, which completes the first key step of verifying the existence of an optimal indemnity $I^*(z) = z - H^*(z)$ to the insurer's problem.
In the second step, we provide two \emph{sufficient} and \emph{necessary} conditions to characterize the optimal retention $H^*$: in the first characterization, $H^*$ satisfies an integral inequality that holds for all admissible retention functions (see Theorem \ref{thm:op1}); in the second characterization, $H^*$ is a solution to a fully nonlinear ordinary integro-differential equation (OIDE) in \eqref{vi001}, which a double-obstacle problem (see Theorem \ref{thm:op2}). 
In the last step, we consider three examples in which the insurer's optimal contracts 
are obtained in (semi)closed form. In the first example, we investigate the case of no distortion (i.e., $\wg(p) = p$ for all $p \in [0,1]$) and find that the optimal reinsurance contract is of stop-loss type (see Corollary \ref{coro2}). In the second example, we consider a special loss distribution $F$ and a special distortion $\wg$, among other modeling conditions, and show that the optimal contract is of multi-layer type and consists of two equally weighted stop-loss contracts (see Proportion \ref{prop:layer}). The last example assumes that the quantile function of $F$, denoted by $F^{-1}$, has a positive derivative. Even for a general distortion, we are able to obtain the optimal indemnity in a semiclosed form (see Proposition \ref{prop:equivalence}).

The main contributions of this paper are discussed as follows. First, we consider a general, not necessarily concave, distortion function $\wg$ in the extended distortion premium principle. 
The relaxation of nonconcavity makes an essential difference in analysis. 
Under a nonconcave distortion $\wg$, the optimal retention function $H^*$ is characterized as a solution to a double-obstacle problem; 
however, if additional concavity is assumed on $\wg$, one may apply the concave envelope to reduce the double obstacle problem into a single obstacle problem (see Corollary \ref{coro2} for an example). 
Second, we show by delicate analysis that the value function $V$ is exponentially decreasing, as in $V(x) = \ee^{-a^* x}$, and the rate $a^*>0$ is unique and can be obtained easily via solving an algebraic equation (see \eqref{eq:v_star}). 
Third, we obtain an easy-to-verify condition, both sufficient and necessary, for a stop-loss reinsurance to be optimal for the insurer (see Corollary \ref{coro1}). In contrast, the extant literature is only able to provide some sufficient conditions for such an optimal contract.

The literature on optimal (re)insurance is rich, and, as already mentioned, many previous works have studied a similar optimal reinsurance problem but \emph{without} the IC constraint (see, e.g., \cite{schmidli2001optimal, gajek2004reinsurance, promislow2005minimizing, bayraktar2007minimizing, bernard2009optimal, bernard2015optimal, liang2020minimizing}). 
However, with and without such a condition could lead to different optimal contracts; see \citet{boonen2019existence}.
In what follows, we focus on those that also impose the IC constraint on indemnities and note that the main difference lies in the premium principles and optimization criteria. (Another noticeable difference is that nearly all those papers consider a static one-period model, while we study under a dynamic setup.)
In terms of premium principles, the majority assumes the expected-value premium principle; see \citet{xu2019optimal}, \citet{chi2020optimal, chi2022regret}, and \citet{xu2021pareto}. 
With regard to optimization criteria, different risk measures (not the ruin probability) are adopted in the literature, such as distortion risk measures in \citet{boonen2022marginal} and general translation invariant risk measures in \citet{boonen2023bowley}; see \citet{cai2020optimal} for a survey article on optimal reinsurance under various risk measures.
Another popular optimization criterion is based on utility theory, including the standard expected utility theory (EUT) in \citet{chi2020optimal}, regret-based EUT in \citet{chi2022regret}, rank-dependent utility theory in \citet{xu2019optimal} and \citet{xu2021pareto, xu2022moral}, and the one with a general increasing utility in \citet{boonen2019existence}.
In comparison, our objective is to minimize the probability of lifetime ruin by considering a much more general distortion premium principle and by allowing the distortion function to be nonconcave.

The reminder of the paper is organized as follows. In Section \ref{sec:model}, we formulate the insurer's optimal reinsurance problem. In Section \ref{sec:sol}, we solve the problem and characterize the optimal retention function. In Section \ref{sec:exm}, we consider three examples in which the optimal reinsurance contract is obtained in (semi)closed form. Technical proofs are placed in Appendix \ref{sec:proof}.

\section{Model}\label{sec:model}

We fix a filtered probability space $(\Omega, \mathcal{F}, \{\mathcal{F}_t\}_{t\geq 0}, \mathbb{P})$ on which all random subjects introduced hereafter are properly defined.
We consider a representative insurer and model its uncontrolled surplus process $\{R(t)\}_{t \ge 0}$ by the classical Cram\'{e}r-Lundberg (CL) model
\begin{align}\label{exp:surplus1}
R(t)=R(0)+\pi t-\sum_{i=1}^{N(t)}Z_i , 
\end{align}
in which $R(0)>0$ is the insurer's initial reserve, $\pi>0$ denotes the constant premium rate, $N=\{N(t)\}_{t \ge 0}$ is a Poisson process modeling claim frequency, and $\{Z_i\}_{i=1,2,\dots}$ is a series of independent and identically distributed random variables, also independent of $N$, with $Z_i$ denoting the $i$-th loss. 
Let $Z \ge 0$ be a generic random variable of loss (i.e., $Z$ and $Z_i$ have the same distribution), and denote its cumulative distribution function (cdf) by $F$. Without loss of generality, we assume the intensity of the Poisson process $N$ is 1 (i.e., $\BE{N(1)} = 1$).

The insurer can access a reinsurance market to purchase per-loss reinsurance contracts to manage its risk exposure. 
Let $I \triangleq I(z)$ and $H \triangleq H(z)$ denote the indemnity function and the retention function of a reinsurance policy, respectively. 
Upon the reporting of a loss $Z = z$, $I(z)$ is the loss amount ceded to the reinsurer and $H(z)$ is the loss amount retained by the insurer. 
It is clear that $I(z) + H(z) = z$ for all $z \ge 0$. 
As discussed in Section \ref{sec:intro}, we only consider reinsurance policies that are moral-hazard-free, and define admissible policies as follows. 

\begin{defn}
\label{def:adm_indem}
A reinsurance policy is admissible if the associated indemnity function $I$ satisfies the following conditions: 
(1) $I:[0, \infty) \mapsto [0, \infty)$, (2) $I(0) = 0$, and (3) $0 \le I( z) - I(z') \le z - z'$ for all $z \ge z' \ge 0$. 
\end{defn}

Define a set $\Ac$ by
\begin{align}
\label{eq:Ac} \quad 
\Ac \triangleq \Big\{f: [0, \infty) \mapsto [0, \infty) \;\Big| f \mbox{ is absolutely continuous, } f(0) = 0, \mbox{ and $0\leq f'\leq 1$ a.e.}
\Big\}. 
\end{align}
We observe an important equivalence result: a reinsurance policy is admissible $\Longleftrightarrow I \in \Ac \Longleftrightarrow\quad H \in \Ac$ (see \citet{xu2019optimal}).
For every $I \in \Ac$, $H = z - I \in \Ac$, and vice versa. As such, the indemnity and retention functions share the same admissible set $\Ac$, and it is equivalent to consider either $I$ or $H$ in the analysis. 
As will become clear, working with admissible retention functions $H \in \Ac$ is more convenient.

\begin{remark}
The third condition in Definition \ref{def:adm_indem} implies that both admissible indemnity and retention functions are increasing, which rules out ex post moral hazard. Such a condition is frequently referred to as the IC constraint, also called the no-sabotage condition. 
When such a constraint is not imposed, the derived optimal reinsurance contract may be strictly decreasing over some region (see, e.g., \citet{bernard2015optimal} for an example).

Mathematically, the third condition in Definition \ref{def:adm_indem} implies that an admissible indemnity $I$ is Lipschitz continuous with a Lipschitz constant of 1, which in turn implies that $I$ is absolutely continuous and is differentiable a.e. with derivative $I' \in [0,1]$. This explains the definition of the admissible set $\Ac$ in \eqref{eq:Ac}.

In Definition \ref{def:adm_indem}, we consider time-invariant indemnities for two reasons. First, the problem we study is stationary, leading to a time-invariant optimal indemnity function (see Remark \ref{rem:stat} below). As such, restricting to time-invariant indemnities does not lose any generality. Second, although dynamic reinsurance policies are well studied in the literature (see \citet{schmidli2001optimal}), static contracts are dominant in real life (see \citet{albrecher2017reinsurance}). 
\end{remark}

To determine the premium rate of a reinsurance policy, we introduce a set $\setw$ 
\begin{align}
\setw \triangleq \Big\{\wg:[0, 1]\to [0, 1]\; \Big|\; & \mbox{$\wg$ is left-continuous and increasing,} \notag \\
&\mbox{ with $\wg(0)=\wg(0+)=0$ and $\wg(1)=1$}\Big\}.
\label{eq:set_w}
\end{align}
Given any $g \in \setw$, define a nonlinear expectation operator for all non-negative variables $\xi$ by\footnote{Throughout this paper, $\int_a^b=\int_{[a,b)}$.} 
\begin{align}
\label{eq:g_exp}
\BE[\wg]{\xi} \triangleq (1+\theta_0) \int_0^{\infty} \wg \Big( \BP{\xi >z} \Big)\dz, 
\end{align}
in which $\theta_0 > -1$ is a constant parameter, and $\wg \in \setw$ is called a distortion function. 
In this work, we assume that the reinsurer applies the extended distortion premium principle to determine the premium rate, as defined below. 

\begin{defn}
\label{def:premium}
The premium rate of a reinsurance policy with an indemnity function $I \in \Ac$ under the extended distortion premium principle is given by 
\begin{align}
\label{eq:premium}
c(I) \triangleq \BE[\wg]{I(Z)} = (1+\theta_0) \int_0^{\infty} \wg \big( \BP{I(Z) >z} \big)\dz,
\end{align}
in which $\wg \in \setw$ is the reinsurer's distortion function, and $\BE[\wg]{\:\cdot\:}$ is defined by \eqref{eq:g_exp}.
\end{defn}

\begin{remark}
\label{rem:premium}
If $\theta_0 = 0$ and $\wg$ is concave, 
\eqref{eq:premium} becomes the standard distortion premium principle; see \citet{wang1996premium} and \citet{wang1997axiomatic} for its properties. 
If $\wg(p)\equiv p \in \setw$, the premium $c$ in \eqref{eq:premium} reduces to the one under the expected value premium principle. 
\eqref{eq:premium} covers many risk-based premium principles (such as {\rm VaR}, {\rm CVaR}, {\rm Mean-CVaR}) and proportional hazard transform principle. We refer readers to Section 2.1 in \citet{boonen2022marginal} for more discussions on \eqref{eq:premium} with $\theta_0 = 0$ and concrete examples of $\wg$ for various premium principles.
\end{remark}

In what follows, we assume that the reinsurer chooses a distortion $\wg \in \setw$ and this decision is known to the insurer. In consequence, once the insurer chooses its reinsurance indemnity $I$, the policy premium is known to the insurer via \eqref{eq:premium}. For any $I \in \Ac$, let $R^I = \{R^I(t)\}_{t \ge 0}$ denote the insurer's controlled surplus process. Although the uncontrolled surplus $R$ ($R = R^0$ with $I(t) \equiv 0$) is given by the CL model in \eqref{exp:surplus1}, we apply the diffusion approximation for $R^I$ (see \citet{grandell1991aspects} for general theory and \citet{schmidli2001optimal} for a similar application in optimal reinsurance problems). 
Under such a diffusion model, the dynamics of $R^I$ follows
\begin{align}
\label{eq:surplus_R}
\dd R^I(t)=\mu(I)\dt+\sigma (I)\dd B(t),
\end{align}
in which $B = \{B(t)\}_{t \ge 0}$ is a standard Brownian motion, 
\begin{align}
\label{eq:mu_sig}
\mu(I) &\triangleq \pi-c(I)-\BE{Z - I(Z)}, 
\quad \text{and} \quad 
\sigma^2(I)\triangleq \BE{(Z-I(Z))^2} = \BE{H^2(Z)}.
\end{align}

Given an admissible indemnity $I \in \Ac$, define the associated \emph{ruin time} $\tau^I$ by 
\begin{align*}
\tau^I \triangleq\inf \big\{t>0 \;\big|\; R^I(t) \le 0 \big\},
\end{align*}
in which $R^I$ follows the dynamics in \eqref{eq:surplus_R}, with the convention that 
$\inf\emptyset=+\infty$. 
Because $R^I$ is continuous, $R^I(\tau^I)=0$ if ruin happens (i.e., $\tau^I < \infty$). 
The insurer's objective is to minimize the probability of lifetime ruin (or ultimate ruin), which is one of the most popular criteria in optimal reinsurance problems (see \citet{schmidli2001optimal} and \citet{promislow2005minimizing} for early contributions). 
We now formulate the insurer's optimal reinsurance problem as follows.

\begin{prob}
The insurer seeks an optimal reinsurance indemnity $I^*$ that minimizes the ruin probability defined by 
\begin{align}
\label{eq:ruin_prob}
V^I (x)\triangleq\BP{\tau^I < \infty \;\Big|\; R^I(t)=x}, \quad x>0.
\end{align}
Equivalently, the insurer solves the following optimization problem:
\begin{align}
\label{valueV}
V(x)\triangleq\inf _{I \in \Ac} \, V^I (x), \quad x > 0.
\end{align}
We call a solution $I^*$ to \eqref{valueV} an optimal indemnity (or an optimal policy) and $V$ the value function of Problem \eqref{valueV}.
\end{prob}

\begin{remark}
\label{rem:stat}
Although the ruin probability in \eqref{eq:ruin_prob} is conditioned at an arbitrary time $t \ge 0$, it is easy to see that such a conditional probability is time-invariant. For this reason, we write $V^I(x)$, instead of $V^I(t,x)$, for the ruin probability in \eqref{eq:ruin_prob}; the same also applies to the notation of the value fuction $V(x)$. Because of the stationarity of the problem, the optimal indemnity $I^*$, if exists, must be time-invariant as well, even when time-dependent dynamic indemnities $I(t,z)$ are allowed. As such, without of loss of generality, we take $t = 0$ in \eqref{eq:ruin_prob} in the subsequent analysis.
\end{remark}

To rule out trivial cases, 
we impose the following assumption in the rest of this paper. 
\begin{assmp}\label{ass:p1} 
Both the insurance and reinsurance policies are not cheap and satisfy
\begin{align}
\label{eq:asu}
0 < \BE{Z} < \pi < \BE[\wg]{Z}<\infty.
\end{align}
\end{assmp}
Recall that $Z$ represents the insurable loss, $\pi$ is the insurance premium rate charged by the insurer, and $c(I) = \BE[\wg]{I(Z)}$ is the reinsurance premium rate paid by the insurer. The above assumption then rules out two trivial cases: (1) if $\BE{Z} \ge \pi$, $V^I(x) = 1$ for all $I$; (2) if $\pi \ge \BE[\wg]{Z}$, $V^I(x) = 0$ for $I(z) = z$.

\section{Solution}
\label{sec:sol}

This section presents a complete solution to the insurer's problem in \eqref{valueV}. We outline the main idea and steps of obtaining the solution as follows. 
In Section \ref{sub:value}, we characterize the value function $V$ by analyzing an equivalent problem; in Section \ref{sub:op_policy}, we show that an optimal reinsurance contract exists and derive optimality conditions to characterize such a contract.

\subsection{Characterization of the value function}
\label{sub:value}

We first apply the dynamic programming to analyze the insurer's problem in \eqref{valueV} and obtain a characterization result of the value function in Theorem \ref{thm:HJB}. We omit its proof, as it is standard in control theory.

\begin{thm}
\label{thm:HJB}
The value function $V$ to Problem \eqref{valueV} solves the following Hamilton-Jacobi-Bellman (HJB) equation in the viscosity sense:
\begin{align}\label{eq:HJB}
0=\inf_{I \in \Ac} \left[\mu(I)V'(x) + \frac{1}{2}\sigma^2(I)V''(x) \right],
\end{align}
in which $\mu(I)$ and $\sigma^2(I)$ are defined in \eqref{eq:mu_sig}, subject to two boundary conditions
\begin{align}
\label{eq:boundary}
V(0)=1 \quad \text{and} \quad V(\infty)=0.
\end{align}
\end{thm}

With the help of Theorem \ref{thm:HJB}, our objective is to find a classical solution for the value function $V$ satisfying \eqref{eq:HJB}-\eqref{eq:boundary}. 
To that end, we make an educated guess on the ansatz 
\begin{align}
\label{eq:ansatz}
\hat{V} (x) = \ee^{- a^* x}, \quad x > 0, 
\end{align}
in which $a^*$ is a positive constant yet to be determined. 
Substituting the above ansatz into the HJB \eqref{eq:HJB} yields 
\begin{align}\label{eq:HJB2}
\inf_{I \in \Ac} \; \left\{ \frac{a^*}{2} \sigma^2(I)-\mu(I) \right\} =0.
\end{align}

To study \eqref{eq:HJB2}, we first introduce a new function $\pw$: 
\begin{align}
\label{eq:F_hat}
\pw(z)\triangleq 1-\frac{\wg(1-F(z))}{\wg(1-F(0))},
\end{align}
in which $\wg \in \setw$ is the distortion function chosen by the reinsurer, and $F$ is the cdf of the loss $Z$. 
We call $\pw$ the \emph{dual distribution} of $F$ under the distortion $\wg$. 
By recalling the definition of $\setw$, we see that $\pw$ is right-continuous and increasing, with $\pw(0+)=\pw(0)=0$ and $\pw(\infty)=1$. 
As a result, $\pw$ can be regarded as the cdf of some positive random variable.

\begin{remark}
If $\wg(1-F(0)) = 0$, we have $\wg(1-F(z))=0$ for all $z \geq 0$, by recalling the monotonicity and non-negativity of $\wg$. This then implies $ \BE[\wg]{Z}=0$ by \eqref{eq:g_exp}, which contradicts the assumption in \eqref{eq:asu}. As such, $\wg(1-F(0)) > 0$ holds, justifying the definition of $\pw$ in \eqref{eq:F_hat}.
In the behavioral finance literature, $\pw(z)$ is usually defined as $1- \wg(1-F(z))$. This is consistent with our definition if $F(0)=0$. But in actuarial applications, it is common that $F(0)>0$; to account for this fact, we modify the standard definition to \eqref{eq:F_hat}. 
\end{remark}

Define $v: (0, \infty) \mapsto \mathbb{R}$ by 
\begin{align}\label{opi1}
v(a) \triangleq \inf_{H \in \Ac} \left\{ \int_0^{\infty}\Big[\frac{a}{2} \, H^2(z)+H(z)\Big]\dd F(z)-(1+\theta)\int_0^{\infty}H(z)\dd\pw(z) \right\}, 
\end{align}
in which $H$ denotes the retention function of a reinsurance policy, $\pw$ is defined in \eqref{eq:F_hat}, and 
\begin{align}
\label{eq:theta}
\theta\triangleq (1+\theta_0) \, \wg(1-F(0))-1>-1.
\end{align}
Since $H(0) = 0$ for any $H \in \Ac$, we can change $\int_0^\infty$ to $\int_{(0, \infty)}$ for both integrals in \eqref{opi1}. This allows us to ``ignore'' the zero point $z=0$ in the subsequent analysis, even though the loss $Z$ may have a probability mass at 0 (i.e., $F(0) > 0$).

To solve the reduced HJB equation in \eqref{eq:HJB2}, we present several lemmas.

\begin{lemma}
\label{lem:equi}
The HJB equation in \eqref{eq:HJB2} is equivalent to 
\begin{align}
\label{eq:v_star}
v( a^* )=\pi-(1+\theta) \int_0^{\infty} z\dd\pw(z),
\end{align}
in which $\pi$ is the insurer's premium rate in \eqref{exp:surplus1}, $a^*$ is the same as in \eqref{eq:ansatz}, and $\pw$, $v$, and $\theta$ are defined by \eqref{eq:F_hat}, \eqref{opi1}, and \eqref{eq:theta}, respectively.
\end{lemma}

\begin{proof}
For any function $f$ such that $f' > 0$ and $f(0) = 0$, we have 
\begin{align*} 
\int_0^{\infty} \wg(\BP{f(Z)>z})\dz
&=\int_0^{\infty} \wg(\BP{f(Z)>f(z)}) \dd f(z)
=\int_0^{\infty} \wg(1 - F(z))f'(z)\dz\\
&=\wg(1-F(0))\int_0^{\infty} (1-\pw(z))f'(z)\dz
=\wg(1-F(0))\int_0^{\infty} f(z)\dd\pw(z),
\end{align*}
which, combining with \eqref{eq:g_exp} and \eqref{eq:theta}, implies 
\begin{align*}
\BE[\wg]{f(Z)}&=(1+\theta_0) \int_0^{\infty} \wg(\BP{f(Z)>z})\dz=(1+\theta) \int_0^{\infty} f(z)\dd\pw(z).
\end{align*} 
By a limit argument, one can show that the above holds for any increasing functions $f$ with $f(0)=0$ (the set of such functions is larger than $\Ac$ in \eqref{eq:Ac}). 
In particular, for any $I \in \Ac$, 
\begin{align}
\label{eq:new_c}
c(I)=\BE[\wg]{I(Z)}=(1+\theta) \int_0^{\infty} I(z) \dd\pw(z) < \infty.
\end{align}
The desired equivalence result then follows by recalling the definitions of $\mu$ and $\sigma$ in \eqref{eq:mu_sig}.
\end{proof}

\begin{lemma}
\label{lem:v}
The function $v$ defined by \eqref{opi1} is bounded, continuous, strictly increasing, and concave on $(0, \infty)$. Furthermore, it satisfies 
\begin{align}\label{lim1} 
-(1+\theta) \int_0^{\infty} z\dd\pw(z)\leq \lim_{a\to 0+}v(a) 
\leq \int_0^{\infty} z\dd F(z)-(1+\theta)\int_0^{\infty} z\dd\pw(z) < 0, 
\end{align}
and
\begin{align}\label{lim2} 
\lim_{a\to+\infty}v(a)=0.
\end{align}
\end{lemma}

\begin{proof}
See Appendix \ref{sec:proof} for the proof.
\end{proof}

\begin{lemma}
\label{lem:a_star}
There exists a unique $a^* \in (0, \infty)$ such that \eqref{eq:v_star} holds. 
\end{lemma}

\begin{proof}
By using \eqref{eq:new_c} with $I(z) = z$, the standing assumption in \eqref{eq:asu} can be rewritten as 
\begin{align}
\label{eq:new_asu}
0 < \int_0^\infty \, z \dd F(z) < \pi < (1 + \theta) \int_0^\infty \, z \dd \pw(z),
\end{align}
which, along with Lemma \ref{lem:v}, proves the assertion.
\end{proof}

With the above preparations, we come to the conclusion that the ansatz in \eqref{eq:ansatz} is a candidate for the value function $V$ of Problem \eqref{valueV}, pending on standard verification arguments, as unfolded in the next theorem.

\begin{thm}
\label{thm:veri}
Let $a^*$ be the unique solution to \eqref{eq:v_star} established in Lemma \ref{lem:a_star}. We have
\begin{align}
V(x) \ge \hat{V}(x) = \ee^{- a^* x}, \quad x > 0.
\end{align}
Moreover, if the infimum problem in $v(a^*)$, defined by \eqref{opi1} with $a = a^*$, admits an optimizer $H^* \in \Ac$, then 
\begin{align}
\label{eq:V_equal}
V(x) = V^{I^*}(x) = \ee^{- a^* x}, \quad x > 0,
\end{align}
in which $I^*$, with $I ^*(z) \triangleq z - H^*(z)$ for all $z \ge 0$, is an optimal indemnity to Problem \eqref{valueV}. 
\end{thm}

\begin{proof}
For any admissible $I \in \Ac$, define stopping times $\{ \nu_n \}_{n=1,2,\cdots}$ by 
\[\nu_n \triangleq \min\Big\{\tau^I, \quad\inf\Big\{t > 0: \int_0^{t}\big(|\mu(I)|^2 \ee^{-2 a^* R^I(s)}+1\big)\ds>n \Big\}\Big\}.\]
Observe that $\nu_n<\infty$ and $\nu_n \to \tau^I $ as $n\to\infty$. 
By applying It\^o's lemma to $\exp (-a^* R^I(\cdot) )$ and using \eqref{eq:HJB2}, we obtain
\begin{align*} 
\ee^{-a^* R^I(\nu_n)}
&\geq \ee^{-a^* R^I(0)} - a^* \int_0^{\nu_n} \mu(I) \, \ee^{- a^* R^I(s)}\dd B(s).
\end{align*}Applying \dct to the above inequality gives 
\begin{align*}
	\ee^{-a^* x} & \le \lim_{n \to \infty} \, \mathbb{E} \left[ \ee^{-a^* R^I(\nu_n)}\;\big|\; R^I(0) = x \right] = \lim_{n \to \infty} \, \mathbb{E} \left[ \ee^{-a^* R^I(\nu_n)} \, \mathrm{1}_{\nu_n < \infty}\;\big|\; R^I(0) = x \right] \\
	&= \mathbb{E} \left[ \ee^{-a^* R^I(\tau^I)} \, \mathrm{1}_{\tau^I < \infty}\;\big|\; R^I(0) = x \right] 
	=\mathbb{P} \left[ \tau^I < \infty\;\big|\; R^I(0) = x \right] = V^I(x)
\end{align*}
for all $I \in \Ac$, which implies $V(x) \ge \hat{V}(x) = \ee^{- a^* x}$.
Since the above inequality becomes an equality at $I^*$, and $\ee^{- a^* x}$ satisfies all the conditions in Theorem \ref{thm:HJB}, the second result then follows. 
\end{proof}

By Lemma \ref{lem:v}, $v(a) < 0$ for all $a \in (0, \infty)$; by the definition of $v$ in \eqref{opi1}, when taking $H \equiv 0 \in \Ac$, the right hand side in \eqref{opi1} is 0. 
As a result, $H(z) \equiv 0$ is \emph{not} optimal, or equivalently, the full reinsurance with $I(z) = z$ is \emph{not} optimal to the insurer.

\subsection{Characterization of the optimal contracts}
\label{sub:op_policy}

The analysis in Section \ref{sub:value} provides a complete characterization of the value function $V$ in \eqref{valueV}, but says nothing about optimal indemnity and retention functions. However, thanks to Lemma \ref{lem:equi}, 
we do know that the original insurer's problem in \eqref{valueV} and the following problem
\begin{align}
\label{eq:new_prob}
v(a^*) = \inf_{H \in \Ac} \left\{ \int_0^{\infty}\Big[\frac{a^*}{2} \, H^2(z)+H(z)\Big]\dd F(z)-(1+\theta)\int_0^{\infty}H(z)\dd\pw(z) \right\}
\end{align}
have the same optimal solution(s). 
The next theorem presents an existence result to the problem in \eqref{opi1}, a more general version of Problem \eqref{eq:new_prob}. This result, along with Theorem \ref{thm:veri}, immediately verifies that the value function $V$ of the insurer's problem is indeed given by \eqref{eq:V_equal}. 
To improve readability, we place length technical proofs of this subsection in Appendix \ref{sec:proof}.

\begin{thm} 
\label{thm:existence}
There exists at least one optimal solution to the infimum problem in \eqref{opi1}.
\end{thm}

Generally speaking, we do not have the uniqueness result regarding the solution to Problem \eqref{opi1}. 
This is because the objective in \eqref{opi1} depends on the supports of both $\dd F$ and $\dd\pw$. 
If the support of $\dd F$ is $[0, \infty)$, then the optimal solution to Problem \eqref{opi1} is unique, by virtue of the strict convexity of square functions. 
Since we do not impose any assumption on the loss variable, the support $\dd F$ may be a true subset of $[0, \infty)$. For instance, if $Z$ is a discrete random variable, the support of $\dd F$ is the collection of all discrete values $Z$ may take, and the uniqueness claim would fail in such a case.
However, even if Problem \eqref{opi1} admits more than one optimal solution, all of them will agree on the joint support of $\dd F$ and $\dd\pw$.

With Theorems \ref{thm:veri} and \ref{thm:existence}, the insurer's optimal reinsurance problem in \eqref{valueV} is fully analyzed. 
Observe that, by taking $a = a^*$ in the proof of Theorem \ref{thm:existence}, $\barh$ defined in \eqref{eq:H_bar} is an optimal retention $H^*$ for the insurer. 
In the rest of this section, we aim to obtain more useful characterizations of $H^*$; recall that, given $H^*$, $I^*(z) = z - H^*(z)$ is optimal to Problem \eqref{valueV}.

\begin{thm}[Optimality condition I]
\label{thm:op1}
An admissible retention function $H^* \in \Ac$ is optimal to Problem \eqref{eq:new_prob} if and only if it satisfies
\begin{align} \label{eq:op_I}
\quad \int_0^{\infty}\big(H(z)- H^*(z)\big)\big[(a^* H^*(z)+1)\dd F(z)-(1+\theta)\dd\pw(z)\big] \geq 0, \quad \forall \, H \in \Ac.
\end{align} 
\end{thm}

Theorem \ref{thm:op1} provides an optimality condition for Problem \eqref{eq:new_prob}, which is equivalent to the original insurer's problem in \eqref{valueV}. It is easily seen from the proof of Theorem \ref{thm:op1} (in Appendix \ref{sec:proof}) that replacing $a^*$ by $a$ in \eqref{eq:op_I} yields an optimality condition for the infimum problem in \eqref{opi1}.
Although the optimality condition in \eqref{eq:op_I} is both necessary and sufficient, applying it to find an optimal retention $H^*$ is still difficult, if not impossible, since one needs to test all admissible retention functions. 
To overcome this setback, we provide an alternative and, more useful, characterization below. 
In contrast to \eqref{eq:op_I} being satisfied by all $H \in \Ac$, the alternative condition in \eqref{vi001} is easier to check as it only involves $H^*$ itself.

\begin{thm}[Optimality condition II] 
\label{thm:op2}
An admissible retention function $H^* \in \Ac$ is optimal to Problem \eqref{eq:new_prob} if and only if it satisfies the following OIDE:
\begin{align} \label{vi001}
\min\Big\{\max\big\{ (H^*(z))' -1, \; \opl(z; H^*)\big\}, \;(H^*(z))' \Big\}=0, \quad\mbox{ a.e.}\;z>0, 
\end{align} 
in which $\Phi$ is defined by 
\begin{align} \label{opldef}
\opl(z; H)&\triangleq\int_z^{\infty}\big[(a^* H(x)+1)\dd F(x)-(1+\theta)\dd\pw(x)\big].
\end{align}
\end{thm}

We next apply Theorem \ref{thm:op2} to derive the condition under which stop-loss reinsurance contracts are optimal.

\begin{coro}\label{coro1}
A stop-loss reinsurance contract with indemnity $I^* (z) = (z - d)^+$, in which $d > 0$ is called the deductible, is optimal to the insurer's problem in \eqref{valueV} if and only if 
\begin{align} \label{deductible}
\qquad
\begin{cases}
-a^*\int_z^{d} (d-x)\dd F(x)+(a^* d+1)(1-F(z))\leq (1+\theta_0)\wg(1-F(z)), & \quad z\in[0, d), \medskip\\
(a^* d+1)(1-F(z))\geq (1+\theta_0)\wg(1-F(z)), & \quad z \in [ d, \infty), 
\end{cases} 
\end{align} 
in which $a^*$ is given explicitly by 
\begin{align} 
\label{eq:a_star}
a^* =\frac{\pi+\int_0^{\infty} (z-d)^+\dd\;( F(z)+(1+\theta_0)\wg(1-F(z)))-\int_0^{\infty}z\dd F(z)}{\frac{1}{2}\int_0^{\infty}\min\{d, z\}^2\dd F(z)}>0.
\end{align} 
\end{coro}

\section{Examples}
\label{sec:exm}

In this section, we consider three examples which lead to more explicit results on the optimal reinsurance contract.

\subsection{The case of no distortion}

In the first example, we assume that the reinsurer does not apply distortion in determining the reinsurance premium. Equivalently, we set $\wg(p) \equiv p$ for all $p \in [0,1]$.
Under such an assumption, the general distortion premium principle in \eqref{eq:premium} reduces to the expected value principle, i.e., $\BE[\wg]{I(Z)} = (1 + \theta_0) \BE{I(Z)}$, with $\theta_0 > 0$. 
We apply Corollary \ref{coro1} to find the optimal reinsurance contract as follows.

\begin{coro}
\label{coro2}
Suppose $\wg(p) \equiv p$ for all $p \in [0,1]$. Then, the optimal reinsurance contract to the insurer's problem in \eqref{valueV} is a stop-loss contract.
\end{coro}

\begin{proof}
Define a function $\psi: \R_+ \times \R_+ \mapsto \R_+ $ by $\psi(d, z) \triangleq \big(2-\min\{1, d^{-1}z\} \big)\min\{d, z\}$. 
For any given $z >0$, $\psi(d, z)$ is nonnegative and increasing as a function of $d$.
By \mct, we obtain 
\begin{align*}
\lim_{d\to 0}\frac{1}{2}\theta_0 \, \int_0^\infty \psi(d, z) \dd F(z) =0 \quad \text{and} \quad \lim_{d\to \infty}\frac{1}{2}\theta_0 \, \int_0^\infty \psi(d, z) \dd F(z) = \theta_0 \BE{Z}.
\end{align*}
Note from \eqref{eq:asu} that $0 < (1+\theta_0)\BE{Z}-\pi < \theta_0 \BE{Z}$. As such, there exists a $d^* > 0$ such that
\begin{align}
\label{eq:d_star}
\frac{1}{2}\theta_0 \, \int_0^\infty \psi(d^*, z) \dd F(z) = (1+\theta_0)\BE{Z}-\pi.
\end{align}

Assume for a moment that a stop-loss contract with indemnity $I^*(z) = (z - d^*)^+$ is optimal to Problem \eqref{valueV}. Then, by recalling $\wg(p) = p$ and \eqref{eq:d_star}, we use \eqref{eq:a_star} to compute $a^*$ by
\begin{align*} 
a^* 
&=\frac{\pi+\theta_0\int_0^{\infty} \min\{d^*, z\}\dd\; F(z)-(1+\theta_0)\int_0^{\infty}z\dd F(z)}{\frac{1}{2}\int_0^{\infty}\min\{d^*, z\}^2\dd F(z)} = \frac{\theta_0}{d^*}.
\end{align*}
Under the above $a^*$ and by replacing $d$ by $d^*$, we verify that the two conditions in \eqref{deductible} hold. As such, by Corollary \ref{coro1}, $I^*(z) = (z - d^*)^+$ is the insurer's optimal reinsurance indemnity. 
\end{proof}

Corollary \ref{coro2} provides a full characterization for the optimal stop-loss reinsurance contract. 
Recall that the reinsurer adopts the expected value principle with loading $\theta_0$ in this example; to further compute $d^*$, let us assume that the insurer also charges premium by the expected value principle, but with loading $\tilde{\theta}$ ($0 < \tilde{\theta} < \theta_0$ by \eqref{eq:asu}). Upon scaling, we set $\BE{Z} = 1$. As an immediate result, $(1+\theta_0)\BE{Z}-\pi = \theta_0 - \tilde{\theta}$. 
By \eqref{eq:d_star}, $d^*$ solves the following equation of $d$: 
\begin{align}
\int_0^d (2 - z/d) z \, \dd F(z) + d (1 - F(d)) = 2 (1 - \kappa), \quad \kappa \triangleq \tilde{\theta} / \theta_0 \in (0,1). 
\end{align}
Once the loss distribution $F$ is known or estimated, a nonlinear solver can easily find $d^*$ for any $\kappa \in (0,1)$, since the left hand side is an increasing function of $d$.

\subsection{An example of optimal multi-layer contract}

In the second example, we consider a special piece-wise loss distribution and impose certain conditions on the reinsurer's distortion $\wg$, under which the optimal reinsurance contract is of multi-layer type, as shown in Proposition \ref{prop:layer}.
The detailed setup of this example is introduced as follows. 
\begin{itemize}
\item The loss distribution $F$ is a mixture of exponential distributions and jumps, with 
$F(z) = 1 - \ee^{-z/6}$, if $z\in [0, 1)$; $F(z) = 1 - \ee^{-z/5}$, if $z\in [1, 6)$; and $F(z) = 1 - \ee^{-z/3}$, if $z \in [6, \infty)$.

\item The distortion function $\wg \in \setw$ satisfies the following extra conditions: 
$\wg(\ee^{-z/3})\leq \ee^{-z/3}$, if $z\in [6, \infty)$;
$\wg(\ee^{-z/5})= \ee^{-z/5}$, if $z\in [4, 6)$;
$\wg(\ee^{-z/5})= \frac{1}{8} \ee^{-z/5}(z+9)-\frac{5}{8} \ee^{-4/5}$, if $z\in [2, 4)$; 
$\wg(\ee^{-z/5})\geq \frac{7}{4} \ee^{-1/5}-\frac{5}{8} \ee^{-2/5}-\frac{5}{8} \ee^{-4/5}$, if $z\in[1, 2)$; 
$\wg(\ee^{-z/6})\geq \frac{7}{4}-\frac{3}{2} \ee^{-1/6}+\frac{5}{4} \ee^{-1/5}-\frac{5}{8} \ee^{-2/5}-\frac{5}{8} \ee^{-4/5}$, if $ z\in [0, 1)$.
We take the special values (or bounds) as above so that $\wg$ can be selected to meet the conditions of $\setw$.

\item Set $\theta_0 = 3$. By \eqref{eq:theta}, we compute $\theta=(1+\theta_0)\wg(1-F(0))-1=4\wg(1)-1=3$.

\item The insurer's premium rate $\pi$ is given by 
\begin{align}
\pi &=\int_0^{1}(z+1) \ee^{-z/6}\dz+\int_1^{2}(z+1) \ee^{-z/5}\dz +\frac{1}{4}\int_2^{4}(z+4) \ee^{-z/5}\dz\nonumber\\ 
&\quad\;+2\int_2^{4}\wg(\ee^{-z/5})\dz +4\int_4^{6} \ee^{-z/5} \dz+4\int_6^{\infty} \wg(\ee^{-z/3})\dz. 
\label{eq:pi_layer}
\end{align}
\end{itemize}

\begin{prop}
\label{prop:layer}
Under the above setting, the optimal indemnity function $I^*$ is given by 
\begin{align}
\label{eq:I_layer}
I^*(z)=\begin{cases}
z-3, &\quad z\in [4, \infty);\\
\frac{1}{2}z-1, &\quad z\in [2, 4);\\
0, &\quad z\in [0, 2).
\end{cases}
\end{align} 
Furthermore, the value function $V(x)$ defined by \eqref{valueV} is equal to $\ee^{-x}$.
\end{prop}

By simple algebra, we can rewrite the optimal indemnity $I^*$ in \eqref{eq:I_layer} as
\[I^*(z)=\frac{1}{2}(z-2)^++\frac{1}{2}(z-4)^+, \quad z \ge 0.\]
As such, the optimal reinsurance can be regarded as a mixture of two stop-loss contracts with equal weights.

Proposition \ref{prop:layer} shows that the optimal reinsurance contract has three distinctive layers: the top layer is of stop-loss type, the middle layer is a mixture of proportional and stop-loss types, and the bottom layer is no reinsurance. 
In the literature, an optimal multi-layer reinsurance contract may arise from different setups, such as under the regret-theoretical EUT in \citet{chi2022regret}.

\begin{proof}[Proof of Proposition \ref{prop:layer}] 
We first note that the insurer's premium $\pi$ is chosen as in \eqref{eq:pi_layer} so that $a^* = 1$; recall that $a^*$ is the unique solution to \eqref{eq:v_star} established in Lemma \ref{lem:a_star}. (We will formally verify \eqref{eq:v_star} holds at $a^* = 1$ later.) 
Then, it suffices to show that $H^*(z)= z - I^*(z)$, with $I^*$ in \eqref{eq:I_layer}, 
is an optimal solution to \eqref{eq:new_prob} with $a^* = 1$. By Theorem \ref{thm:op2} and Lemma \ref{lem:in_eq}, it is equivalent to show that the following holds for the above $H^*$: 
\begin{align}
\label{eq:H_layer}
\opl(z; H^*)\geq 0, \; z\in [4, \infty); 
\quad 
\opl(z; H^*)= 0, \; z\in [2, 4);
\quad
\opl(z; H^*)\leq 0, \; z\in [0, 2).
\end{align}
Recall $\Phi(z; H^*)$ is defined in \eqref{opldef} and can be reduced into 
\begin{align*}
\opl(z; H^*) = a^* \int_z^{\infty} (1-F(x)) (H^*)'(x)\dx+(1-F(z))(a^* H^*(z)+1)-(1+\theta_0)\wg(1-F(z)).
\end{align*}

To prove the above claim, we consider four separate regions of $z$ and compute $\Phi(z; H^*)$ using the above result. 
\begin{enumerate}[(1)]
\item For all $z\geq 4$, 
\begin{align*} 
\opl(z; H^*) 
&=4(1-F(z))-4\wg(1-F(z))
=\begin{cases}
4 \, \ee^{-z/3}-4\wg(\ee^{-z/3})\geq 0, &\quad z\in[6, \infty);\\
4 \, \ee^{-z/5}-4\wg(\ee^{-z/5})= 0, &\quad z\in[4, 6).
\end{cases}
\end{align*} 
\item For all $z\in [2, 4)$, 
\begin{align*} 
\opl(z; H^*) 
&= \frac{1}{2}\int_z^{4} \ee^{-x/5}\dx+\frac{1}{2} \, \ee^{-z/5}(z+4)-4\wg(\ee^{-z/5})\\
&= \frac{5}{2}(\ee^{-z/5}- \ee^{-4/5})+\frac{1}{2} \ee^{-z/5}(z+4)-4\wg(\ee^{-z/5})= 0.
\end{align*} 
\item For all $z\in [1, 2)$, 
\begin{align*} 
\opl(z; H^*) 
&= \int_z^{2} \ee^{-x/5}\dx+\int_2^{4} \frac{1}{2} \, \ee^{-x/5}\dx+ \ee^{-z/5}(z+1)-4\wg( \ee^{-z/5})\\
&= 5(\ee^{-z/5}- \ee^{-2/5})+\frac{5}{2}(\ee^{-2/5}- \ee^{-4/5})+ \ee^{-z/5}(z+1)-4\wg(\ee^{-z/5})\\
&=4\Big[\frac{1}{4} \ee^{-z/5}(z+6)-\frac{5}{8} \ee^{-2/5}-\frac{5}{8} \ee^{-4/5}-\wg(\ee^{-z/5})\Big]
\le 0, 
\end{align*} 
in which we have used the fact that $\ee^{-z/5}(z+6)$ is decreasing over $[1,2)$ to derive the inequality.
\item For all $z\in [0, 1 )$, 
\begin{align*} 
\opl(z; H^*) 
&= \int_z^{1} \ee^{-x/6}\dx+\int_1^{2} \ee^{-x/5}\dx+\int_2^{4} \frac{1}{2} \ee^{-x/5}\dx+ \ee^{-z/6}(z+1)-4\wg(\ee^{-z/6})\\
&=4\Big[\frac{1}{4} \ee^{-z/6}(z+7)-\frac{3}{2} \ee^{-1/6}+\frac{5}{4} \ee^{-1/5}-\frac{5}{8} \ee^{-2/5}-\frac{5}{8} \ee^{-4/5}- \wg(\ee^{-z/6})\Big] \le 0. 
\end{align*} 

\end{enumerate}

Since all conditions in \eqref{eq:H_layer} hold as we wish, the proof will be complete once we verify that $a^* = 1$. To that end, we use the definition of $v$ in \eqref{opi1} to compute $v(1)$ and next calculate $(1+\theta)\int_0^{\infty}z\dd\pw(z)$, which, together with \eqref{eq:pi_layer}, give $v(1) = \pi - (1+\theta)\int_0^{\infty}z\dd\pw(z)$. 
By \eqref{eq:v_star} and Lemma \ref{lem:a_star}, $a^*$ must be 1, as claimed.
The proof is complete.
\end{proof}

\subsection{An example with differentiable quantile of the loss distribution} 
\label{sub:quantile}

For the loss distribution $F$, define its quantile function $F^{-1}$ by 
\begin{align*}
	F^{-1}(p) \triangleq \inf \left\{ z \ge 0 | F(z) \ge p \right\}, \quad p \in (0,1],
\end{align*}
with the convention that $\inf \emptyset = + \infty$, 
and set $F^{-1}(0) = \lim_{p \to 0+} \, F^{-1}(p)$. 
In this example, we impose the following assumptions on $F^{-1}$. 
\begin{assmp}\label{ass1} 
The quantile function $F^{-1}$ is absolutely continuous over $[0,1]$, with $F^{-1}(0)=0$, and its derivative satisfies $h(p) \triangleq (F^{-1})'(p)> 0$ for a.e. $p\in(\ms, 1)$, in which $\ms \triangleq F(0)<1$. 
\end{assmp} 
If $\ms>0$, then $Z$ has a positive mass at 0, so \citeassmp{ass1} covers the most common and important case with losses having a positive mass at 0 in insurance. 
Under \citeassmp{ass1}, the loss distribution $F$ is strictly increasing when its value is strictly less than $1$. Also, $F^{-1}(F(x))=x$ for all $0\leq x\leq \esup Z$, $F^{-1}(p)=0$ for $p\leq \ms$, and $F^{-1}(p)>0$ for $p>\ms$. These facts will be used frequently in the subsequent analysis without claim. 

To solve Problem \eqref{valueV}, we introduce the following ODE for $\opll$: 
\begin{align} \label{vi003}
	\begin{cases}
		\min\Big\{\max\big\{\opll''(p)- a^* h(p), \; 1-p-(1+\theta_0)\wg(1-p)-\opll(p)\big\}, \;\opll''(p)\Big\}=0, \quad\\
		\hfill\mbox{ a.e.\; $p\in (\ms, 1)$}, \\
		\opll(1)=0, \quad \opll'(\ms)=0.
	\end{cases}
\end{align} 
Also denote $C^{2-}[m_0, 1]$ as the set of all differentiable functions $f : [m_0, 1] \to \R$, whose derivatives are absolutely continuous.

We now present the key mathematical results of this section in Proposition \ref{prop:equivalence} and discuss its significance in Remark \ref{rem:discussion}. 
\begin{prop}
	\label{prop:equivalence}
	Suppose \citeassmp{ass1} holds. We have the following assertions: 
	\begin{enumerate}[(1).]
		\item If $I^* \in \Ac$ is an optimal indemnity to Problem \eqref{valueV}, 
		then 
		\begin{align}
			\label{eq:opll}
			\opll(p)=-a^* \int_p^1 \left(F^{-1}(x)- I^* \big( F^{-1}(x) \big) \right)\dx,\quad p\in[\ms, 1],
		\end{align}
		is a solution to the ODE in \eqref{vi003} in the class of $C^{2-}([\ms, 1])$ functions.
		
		\item
		If $\opll$ is a solution to the ODE in \eqref{vi003} in the class of $C^{2-}([\ms, 1])$ functions, then
		\begin{align}
			\label{eq:I_op_new} 
			I^*(z)= z-\frac{1}{a^*} \, \opll'(F(z)), \quad z\geq 0,
		\end{align}
		is an optimal indemnity to Problem \eqref{valueV}.
		
		\item The ODE \eqref{vi003} admits a unique solution in $ C^{2-}([\ms, 1])$. 
	\end{enumerate}
\end{prop}

\begin{proof}
(1). Let $I^* \in \Ac$ be an optimal solutions to Problem \eqref{valueV}, and define $\opll$ by \eqref{eq:opll}. We easily see that $\opll\in C^{2-}([\ms, 1])$; furthermore, $\opll(1) = 0$ and $\opll'(\ms)=a^* (F^{-1}(\ms)- I^*(F^{-1}(\ms)))=a^* (0- I^*(0))=0$, implying that the two boundary conditions in \eqref{vi003} are satisfied.

Since $\opll'(F(z))= a^* (z - I^*(z))= a^* H^*(z)$, it follows that $\Phi$ in \eqref{opldef} at $H = H^*$ equals 
\begin{align*}
	\Phi(z; H^*) = -\opll(F(z))+1-F(z)-(1+\theta_0)\wg(1-F(z));
\end{align*}
by Theorem \ref{thm:op2}, $H^*$ solves the OIDE \eqref{vi001}. Using the above results allows us to transform \eqref{vi001} into 
\begin{align*} 
	\min \Bigg \{\max \left\{\frac{1}{a^*}\opll''(F(z))F'(z)-1, \; -\opll(F(z))+1-F(z)-(1+\theta_0)\wg(1-F(z))\right\}, \\
	\hfill\;\frac{1}{a^*}\opll''(F(z))F'(z) \Bigg\}=0, \quad\mbox{ a.e.}\;z>0,
\end{align*} 
which is equivalent to \eqref{vi003} by Lemma 4.4 in \citet{xu2021pareto} and the change of variable $p = F(z)$.

(2). Let $\opll\in C^{2-}([\ms, 1])$ be a solution to the ODE \eqref{vi003}. We define $H^*$ by
\begin{align}
	\label{barhdef1}
H^*(z) 	= \frac{1}{a^*} \, \opll'(F(z)), \quad z \ge 0,
\end{align} 
from which we get 
\begin{align*} 
	\opll(F(z)) &=-a^* \int_{F(z)}^{1} H^*(F^{-1}(t))\dt =-a^* \int_z^{\infty} H^*(x)\dd F(x), \quad z\geq 0.
\end{align*}
By \eqref{opldef}, $\Phi(z; H^*)=-\opll(F(z))+1-F(z)-(1+\theta_0)\wg(1-F(z)).$ 

Using the above identities and \eqref{vi003}, we successfully verify that $H^*$ in \eqref{barhdef1} satisfies \eqref{vi001}.
Therefore, by Theorem \ref{thm:op2}, $I^*$ in \eqref{eq:I_op_new} is an optimal solution to Problem \eqref{valueV}.

(3). Assume to the contrary that $\opll_1$ and $\opll_2$ are two solutions in $ C^{2-}([\ms, 1])$ to the ODE \eqref{vi003}. By assertion (2) above, we can construct two solutions by \eqref{barhdef1}, denoted by $H_1$ and $H_2$, to the infimum problem in \eqref{opi1}. 
However, by the strictly convexity of a square function and $F'>0$ from Assumption \ref{ass1}, the infimum problem in \eqref{opi1} admits a unique solution for each $a > 0$. 
Therefore, $H_1$ and $H_2$ must coincide, implying $\opll'_1(F(z))=\opll'_2(F(z))$ for all $z\geq 0$, or equivalently $\opll'_1(p)=\opll'_2(p)$ for all $p\in[\ms, 1]$. 
In consequence, $\opll_1-\opll_2$ must be a linear function on $[\ms, 1]$. 
By the boundary conditions in \eqref{vi003}, $\opll_1(1)=\opll_2(1)$ and $\opll_1'(\ms)=\opll'_2(\ms)$, we finally claim that $\opll_1-\opll_2$ is identical to zero as desired. 
\end{proof}

Note that the distortion function $\wg$ in this paper is \emph{not} necessarily concave; see the definition of admissible distortions in \eqref{eq:set_w}. However, when $\wg$ is concave, we can apply Proposition \ref{prop:equivalence} to show that an optimal indemnity to Problem \eqref{valueV} can be obtained as a solution to a single-obstacle problem, instead of a double-obstacle problem as in \eqref{vi003}.

\begin{coro}
\label{cor:diff}
Suppose \citeassmp{ass1} holds and $\wg$ is concave on $[\ms, 1]$. 
Then, the optimal indemnity function $I^*$ to Problem \eqref{valueV} is given by 
\begin{align*} 
I^*(z)= -\frac{1}{a^*} \, \Theta'(F(z)), \quad z\geq 0, 
\end{align*}
in which $a^*$ is given by \eqref{eq:v_star}, and $\Theta\in C^{2-}([\ms, 1])$ is the unique solution to the following ODE:
\begin{align} \label{vi003c}
	\begin{cases}
		\max\Big\{\Theta''(p), \; 1-p-(1+\theta_0)\wg(1-p)+ a^* \int_p^1 F^{-1}(t)\dt-\Theta(p) \Big\}=0, \quad\\
		\hfill \mbox{ a.e.\; $p\in (\ms, 1)$}, \\
		\Theta(1)=0, \quad \Theta'(\ms)=-a^* m_0.
	\end{cases}
\end{align}
\end{coro}

\begin{proof}
	Suppose the ODE in \eqref{vi003c} admits a solution $\Theta\in C^{2-}([\ms, 1])$. Using this $\Theta$, we define
	\[\opll(p)=\Theta(p)- a^* \int_p^1 F^{-1}(t)\dt, \quad p\in[\ms, 1].\]
	Then, $\opll\in C^{2-}([\ms, 1])$ and satisfies
	\begin{align} \label{vi003b}
		\begin{cases}
			\max\big\{\opll''(p)-a^* h(p), \; 
			1-p-(1+\theta_0)\wg(1-p)-\opll(p)
			\big\}=0, 
			\hfill \mbox{ a.e.\; $p\in (\ms, 1)$}, \\
			\opll(1)=0, \quad \opll'(\ms)=0.
		\end{cases}
	\end{align} 
	
	We first show that $\opll$ above is convex on $(\ms, 1)$.
	Since $\wg$ is concave by assumption, it is continuous on $(\ms, 1)$. Define a set $K$ by 
	\[K=\big\{p\in(\ms, 1)\;\big|\; \opll(p)>1-p-(1+\theta_0)\wg(1-p)\big\}.\]
	Then, by continuity, $K$ is an open set and thus consists of countably many disjoint open intervals, denoted by $\{K_i\}_{i=1,2,\dots}$. By \eqref{vi003b}, on any such each $K_i$, 
	\[\opll'(\cdot)=\int^{\cdot} \opll''(p)\ddp=a^* \int^{\cdot} h(p)\ddp\]
	is non-decreasing as $h> 0$ from Assumption \ref{ass1}. Since $\opll'$ is continuous, it is non-decreasing on $\overline{K}_i$, the closure of $K_i$.
	For any $p\in(\ms, 1)/\cup_i \overline{K}_i$, we have $\opll(p)=1-p-(1+\theta_0)\wg(1-p)$, by \eqref{vi003b} and continuity. (Notice $(\ms, 1)/\cup_i \overline{K}_i$ is also an open set.) Because $\opll(p)=1-p-(1+\theta_0)\wg(1-p)$ and $\wg$ is concave, $\opll$ is convex on each subinterval of $(\ms, 1)/\cup\overline{K}_i$, and correspondingly $\opll'$ is non-decreasing on the closure of each subinterval. By continuity, we conclude that $\opll'$ is non-decreasing on $(\ms, 1)$. 
	Hence, $\opll$ is convex on $(\ms, 1)$. 
	This together with \eqref{vi003b} implies that $\opll$ satisfies \eqref{vi003}. The other claims follow from Proposition \ref{prop:equivalence} immediately.
\end{proof}

\begin{remark}
	\label{rem:discussion}
Recall that we only consider moral-hazard-free reinsurance contracts whose indemnities satisfy $ 0 \le I' \le 1$ as in \eqref{eq:Ac}.
Because of such double constraints, the ODE \eqref{vi003} is a double-obstacle problem (written in a min-max form), which is notoriously difficult to solve. Previous works often impose strong sufficient conditions to ensure $\opll''\geq 0$, which will help reduce a double-obstacle problem into a single-obstacle problem. \citet{xu2021pareto} is one of the first papers that solve double-obstacle problems without assuming $\opll''\geq 0$. 
Here, we are able to accomplish a similar mission in Proposition \ref{prop:equivalence} under rather general conditions, albeit for a special double-obstacle problem arising from the study of the insurer's optimal reinsurance problem in \eqref{valueV}.
Under the additional assumption that $\wg$ is concave, we further show that the ODE \eqref{vi003} is equivalent to the ODE \eqref{vi003c}, which is a single-obstacle problem.
\end{remark}

\appendix

\section{Technical proofs}
\label{sec:proof}

\begin{proof}[\underline{Proof of Lemma \ref{lem:v}}]

For any $H \in \Ac$, $0\leq H(z)\leq z$ for all $z \ge 0$, which immediately gives the lower bound of $v$ in \eqref{lim1}. Since $H (z) \equiv 0 \in \Ac$, $v(a) \le 0$. 
By recalling the definition of $v$ in \eqref{opi1}, we easily verify that $v$ is bounded, continuous, increasing, and concave. 

In fact, a sharper result $v(a) < 0$ holds. To see this, consider $H_n(z)= n^{-3} (z \wedge n)$; we easily verify that $H_n \in \Ac$ and $H_n \le n^{-2}$ for all $n=1,2,\cdots$. By using such $H_n$ and \eqref{opi1}, we obtain 
\begin{align*} 
v(a) &\leq \int_0^{\infty}\Big[\frac{a}{2}( n^{-3}(z\wedge n))^2+( n^{-3}(z\wedge n))\Big]\dd F(z)-(1+\theta)\int_0^{\infty} ( n^{-3}(z\wedge n))\dd\pw(z)\nn\\ 
&\leq \frac{a}{2}n^{-4}+n^{-3}\int_0^{\infty} z\dd F(z)-(1+\theta)n^{-3}\int_0^{\infty} (z\wedge n)\dd\pw(z)\nn\\ 
&=n^{-3}\Big(\frac{a}{2}n^{-1}+\int_0^{\infty} z\dd F(z)-(1+\theta)\int_0^{\infty} (z\wedge n)\dd\pw(z)\Big),
\end{align*} 
which, upon applying \mct and using \eqref{eq:new_asu}, yields $v(a) < 0$. 

To establish the upper bound of $v$ in \eqref{lim1}, we consider $\tilde{H}_n(z) = z \wedge n$, in which $n > 1$. 
For any $n > 1$, we consider $v$ at $a = n^{-3}$ and use the face that $\tilde{H}_n \in \Ac$ to obtain 
\begin{align*}
v(n^{-3})&\leq 
\int_0^{\infty}\Big[\frac{n^{-3}}{2}(z\wedge n)^2+(z\wedge n)\Big]\dd F(z)-(1+\theta)\int_0^{\infty}(z\wedge n)\dd\pw(z)\\ 
&\leq \frac{1}{2n}+\int_0^{\infty} z\dd F(z)-(1+\theta)\int_0^{\infty} (z\wedge n)\dd\pw(z),
\end{align*}
which, combining with the monotonicity of $v$ and \mct, confirms the upper bound of $v$ in \eqref{lim1}.

In the next agenda, we show the limit result in \eqref{lim2} in two steps.
\\[1ex]
\underline{Step 1}. We only consider distortions $\wg$ that are continuous and piecewise differentiable. 
For such $\wg$, by \eqref{eq:F_hat}, we have 
\[\dd\pw(z)=\frac{\wg'(1-F(z))}{\wg(1-F(0))}\dd F(z).\]
For notational simplicity, we write, for $z>0$, 
\[\beta=\beta(z)=(1+\theta)\frac{\wg'(1-F(z))}{\wg(1-F(0))}\geq 0\]
and 
\[k=k(z)=\sqrt{2az\beta+(1+\beta)^2}-1-\beta\geq 0.\]
It is easy to check 
\[-\frac{k}{a}\beta=\frac{k^2+2k}{2a}-z\beta.\]
We proceed to compute the integrand in \eqref{opi1} and obtain 
\begin{align*} 
\frac{a}{2} H^2(z) +H(z)-H(z) (1+\theta)\frac{\wg'(1-F(z))}{\wg(1-F(0))}
\geq -H(z) (1+\theta)\frac{\wg'(1-F(z))}{\wg(1-F(0))}\geq-\frac{k}{a}\beta,
\end{align*} 
when $0\leq H(z)\leq \frac{k}{a}$; and 
\begin{align*} 
\frac{a}{2}H^2(z)+H(z)-H(z) (1+\theta)\frac{\wg'(1-F(z))}{\wg(1-F(0))}
\geq \frac{a}{2}\left(\frac{k}{a}\right)^2+\frac{k}{a}-z\beta=-\frac{k}{a}\beta,
\end{align*} 
when $\frac{k}{a}\leq H(z)\leq z$.
Hence, we have
\begin{align*} 
v(a)&=\inf_{H \in \Ac} \int_0^{\infty}\Big[\frac{a}{2}H(z)^2+H(z)-(1+\theta) H(z) \frac{\wg'(1-F(z))}{\wg(1-F(0))}\Big]\dd F(z) \\ 
&\geq-\int_0^{\infty} \frac{k(z)}{a}\beta(z)\dd F(z)=-\int_0^{\infty}\frac{2z\beta^2(z)}{\sqrt{2az\beta(z)+(1+\beta(z))^2}+1+\beta(z)}\dd F(z).
\end{align*} 
Notice 
\[0\leq \frac{2z\beta^2(z)}{\sqrt{2az\beta(z)+(1+\beta(z))^2}+1+\beta(z)}
\leq \frac{2z\beta^2(z)}{\sqrt{(1+\beta(z))^2}+1+\beta(z)}\leq z \beta(z)\]
and, by virtue of \eqref{eq:new_c},
\begin{align*} 
\int_0^{\infty} z\beta(z)\dd F(z)=(1+\theta)\int_0^{\infty} z \dd\pw(z)<\infty. 
\end{align*} 
The above results, along with the \dct, imply $\lim_{a \to +\infty} \, v(a) \ge 0$.
Finally, by recalling $v(a) < 0$ for all $a \in (0, \infty)$, \eqref{lim2} holds.
\\[1ex]
\underline{Step 2}. We now consider general distortions $\wg \in \setw$. 
In this case, let $g_0$ be the concave envelope of $\wg$ on $[0, 1-F(0)]$ and smoothly increase to 1 on $(1-F(0),1]$. 
By such a construction, $g_0 \in \setw$ is continuous and piecewise differentiable, fitting the conditions of $\wg$ in Step 1. 
Let $\pw_0$ be the dual for $F$ under the distortion $g_0$, as defined by \eqref{eq:F_hat}. 
Since $g_0(1-F(0))=\wg(1-F(0))$ and $g_0\geq \wg$ on $[0,1-F(0)]$, $\pw\geq \pw_0$. 
For any $H \in \Ac$ (implying $H' \ge 0$), Fubini's theorem yields 
\begin{align*} 
\int_0^{\infty}H(z)\dd \left(\pw_0(z)-\pw(z) \right)=\int_0^{\infty} \left(\pw(z)-\pw_0(z) \right) \, H'(z)\dz \geq 0. 
\end{align*}
The analysis in Step 1 shows that 
\begin{align*}
\lim_{a\to+\infty}\inf_{H \in \Ac} \int_0^{\infty}\Big[\frac{a}{2} \, H^2(z)+H(z)\Big]\dd F(z)-(1+\theta)\int_0^{\infty}H(z)\dd\pw_0(z)=0.
\end{align*}
Combining these two results implies $\lim_{a\to+\infty}v(a) \ge 0$, which, upon recalling $v(a) < 0$, verifies \eqref{lim2}. 

We have already seen that $v$ is increasing; in the remaining of the proof, we further show that $v$ is strictly increasing over $(0, \infty)$. 
To that end, assume that the contrary is true. 
As a consequence, the maximum of $v$ is achieved at an internal point $a_M \in (0, \infty)$ by the concavity of $v$. Since $v(a) < 0$ for all $a \in (0, \infty)$, we must have $v(a_M) < 0$, which contradicts the limit result $\lim_{a \to \infty} \, v(a) = 0$ in \eqref{lim2}. Therefore, $v$ must be strictly increasing over $(0, \infty)$.
\end{proof} 

\begin{proof}[\underline{Proof of Theorem \ref{thm:existence}}]
We use a compact argument to prove the claim.
Let $\{H_n\}_{n=1,2,\cdots}$ be a minimizing sequence for the infimum problem in \eqref{opi1}. 
Since $0\leq H'_n\leq 1$, by the Arzel\`a-Ascoli theorem, for each positive integer $M$, there exists a positive integer $S_M$, a continuous function $\barh_M$, and a subsequence, which we still as $\{H_n\}_{n=1,2,\cdots}$, such that 
\[|H_{n}(z)-\barh_M(z)|< \ee^{-M}, \quad\mbox{ for all $z\in [0, M]$ and $n\geq S_M$}.\]
By the triangle inequality, the sequence $\{H_{S_{M}}\}_{M=1,2,\cdots}$ satisfies 
\begin{align}\label{convergence1} 
|H_{S_{M}}(z)-H_{S_{M+1}}(z)|\leq |H_{S_{M}}(z)-\barh_M(z)|+|H_{S_{M+1}}(z)-\barh_M(z)|<2 \ee^{-M}, 
\end{align}
for all $z\in [0, M]$. 
Therefore, the sequence $\{H_{S_{M}}\}_{M=1,2,\cdots}$ converges to a limit $\barh$ 
\begin{align}
\label{eq:H_bar}
\barh(z) = \lim_{M \to \infty} \, H_{S_{M}} (z), \quad z \ge 0.
\end{align}
By \eqref{eq:Ac}, we can verify that $\barh \in \Ac$ is an admissible retention function.

In what follows, we aim to show that $\barh$ defined above is an optimal solution to the infimum problem in \eqref{opi1}. 
For any $\ep\in(0, 1)$, we know from \eqref{eq:new_asu} that there exists a constant $K_{\ep}>0$ such that 
\[ (1+\theta)\int_{[K_{\ep}, \infty)} z\dd\pw(z)<\ep.\]
Because $0\leq H_n(z)\leq z$ for all $n$, it follows that 
\[ (1+\theta)\int_{[K, \infty)} H_{S_{M}}(z)\dd\pw(z)\leq (1+\theta)\int_{[K_{\ep}, \infty)} z\dd\pw(z)<\ep\]
for all $K\geq K_{\ep}$ and all $M$. 
By \dct, and noticing that $\{H_{S_{M}}\}$ is a minimizing sequence for the infimum problem in \eqref{opi1}, we deduce
\begin{align*} 
&\quad\; \int_{[0, K)}\Big[\frac{a}{2} \, \barh^2(z) + \barh(z)\Big]\dd F(z)-(1+\theta)\int_{[0, K)} \barh(z)\dd\pw(z)\\ 
&= \lim_{M \to\infty}\bigg[\int_{[0, K)}\Big[\frac{a}{2}\, \big( H_{S_{M}}(z) \big)^2+H_{S_{M}}(z)\Big]\dd F(z)-(1+\theta)\int_{[0, K)} H_{S_{M}}(z)\dd\pw(z)\bigg]\\
&\leq \lim_{M\to\infty}\bigg[\int_0^{\infty}\Big[\frac{a}{2} \, \big( H_{S_{M}}(z) \big)^2 +H_{S_{M}}(z)\Big]\dd F(z)-(1+\theta)\int_0^{\infty} H_{S_{M}}(z)\dd\pw(z)+\ep\bigg]\\
&=v(a)+\ep, 
\end{align*}
for all $K\geq K_{\ep}$. Applying \mct and using $0\leq \barh(z)\leq z$ yield
\[0\leq \lim_{K\to\infty} \int_{[0, K)} \barh(z)\dd\pw(z)=\int_0^{\infty} \barh(z)\dd\pw(z)\leq \int_0^{\infty}z\dd\pw(z)<\infty, \]
and
\[\lim_{K\to\infty} \int_{[0, K)}\Big[\frac{a}{2} \, \barh^2(z) + \barh(z)\Big]\dd F(z)=\int_0^{\infty} \Big[\frac{a}{2} \, \barh^2(z) + \barh(z)\Big]\dd F(z), \]
from which we get 
\begin{align*} 
\int_0^{\infty}\Big[\frac{a}{2} \, \barh^2(z) + \barh(z)\Big]\dd F(z)-(1+\theta)\int_0^{\infty} \barh(z)\dd\pw(z)&\leq v(a)+\ep.
\end{align*}
Since $\ep>0$ is arbitrarily chosen, this shows $\barh$ is an optimal solution to the infimum problem in \eqref{opi1}.
\end{proof}

\begin{proof}[\underline{Proof of Theorem \ref{thm:op1}}] The proof is divided into two steps.
\\[1ex]
\underline{Step 1}. Suppose $H^* \in \Ac$ is an optimal solution to Problem \eqref{eq:new_prob}, and we show \eqref{eq:op_I} holds. 

To that end, for any $H\in \Ac$, $\ep \in (0,1)$, and $n>1$, define 
\[H_{\ep}(z)= H^*(z)+\ep \big( H(z)\wedge n - H^*(z) \big).\]
Then, it is easily seen that $H_{\ep} \in \Ac$. 
Using the optimality of $H^*$, we get 
\begin{align*} 
0 &\leq \liminf_{\ep\to 0+}\frac{1}{\ep}\bigg[\Big\{\int_0^{\infty}\Big[\frac{a^*}{2} (H_{\ep}(z))^2+H_{\ep}(z)\Big]\dd F(z)-(1+\theta)\int_0^{\infty}H_{\ep}(z)\dd\pw(z)\Big\}\\
&\qquad\qquad\qquad -\Big\{\int_0^{\infty}\Big[\frac{a^*}{2} (H^*(z))^2+ H^*(z)\Big]\dd F(z)-(1+\theta)\int_0^{\infty} H^*(z)\dd\pw(z)\Big\}\bigg]\\
&=\liminf_{\ep\to 0+}\bigg[\int_0^{\infty}\big((H(z)\wedge n)- H^*(z)\big)\big[(a^* H^*(z)+1)\dd F(z)-(1+\theta)\dd\pw(z)\big] \\
&\qquad\qquad\qquad +\ep\int_0^{\infty}\frac{a^*}{2}\big((H(z)\wedge n)- H^*(z)\big)^2 \dd F(z)\bigg]\\
&=\int_0^{\infty}\big((H(z)\wedge n)- H^*(z)\big)\big[(a^* H^*(z)+1)\dd F(z)-(1+\theta)\dd\pw(z)\big],
\end{align*} 
in which we have used the fact that $v(a^*) < 0$ from Lemma \ref{lem:v} to derive the last equality. 

Applying \mct and noting $0\leq H(z)\leq z$ lead to 
\begin{align*} 
\lim_{n \to \infty} \int_0^{\infty} (H(z)\wedge n) (a^* H^*(z)+1)\dd F(z) 
=\int_0^{\infty} H(z) (a^* H^*(z)+1)\dd F(z)
\end{align*} 
and
\begin{align*} 
0\leq \lim_{n \to \infty} \int_0^{\infty} (H(z)\wedge n ) \dd\pw(z)=\int_0^{\infty} H(z) \dd\pw(z)
\leq \int_0^{\infty} z\dd\pw(z)<\infty,
\end{align*} 
which together prove \eqref{eq:op_I}.
\\[1ex] 
\underline{Step 2}. Suppose there exists an $H^* \in \Ac$ that satisfies \eqref{eq:op_I}, and we aim to show that such an $H^*$ is an optimal solution to Problem \eqref{eq:new_prob}. 

Since $H^* / 2 \in \Ac$, we replace $H$ by $H^*/2$ in \eqref{eq:op_I} and obtain 
\[\int_0^{\infty} H^*(z) \big[(a^* H^*(z)+1)\dd F(z)-(1+\theta)\dd\pw(z)\big] \leq 0.\]
Because $0\leq H^*(z)\leq z$, we have 
\[ 0\leq \int_0^{\infty} H^*(z) \dd F(z)\leq \int_0^{\infty}z \dd F(z)<\infty \quad 
\text{ and } \quad 0\leq \int_0^{\infty} H^*(z) \dd\pw(z)\leq \int_0^{\infty}z \dd \pw(z)<\infty.\] 
It then follows
\begin{align}\label{L2a}
a^* \int_0^{\infty} (H^*(z))^2\dd F(z)<\infty.
\end{align}

Now suppose that $H^*$ is not an optimal solution to Problem \eqref{eq:new_prob}. 
Then, there exists a function $H_{1}\in\Ac$ and a constant $c>0$ such that 
\begin{align*} 
&\quad\; \int_0^{\infty}\Big[\frac{a^*}{2} (H_{1}(z))^2+H_{1}(z)\Big]\dd F(z)-(1+\theta)\int_0^{\infty}H_{1}(z)\dd\pw(z)\\
&\leq \int_0^{\infty}\Big[\frac{a^*}{2} (H^*(z))^2+ H^*(z)\Big]\dd F(z)-(1+\theta)\int_0^{\infty} H^*(z)\dd\pw(z)-c<\infty.
\end{align*} 
By a similar argument as above, 
\begin{align}\label{L2b}
a^* \int_0^{\infty} (H_1(z))^2\dd F(z)<\infty.
\end{align} 

For $\ep\in(0, 1)$, define $H_{\ep}(z)=\ep H_1(z)+(1-\ep) H^*(z)$, 
and note, by \eqref{eq:Ac}, that $H_{\ep}\in \Ac$. 
By the convexity of square functions and the hypothesis, we have 
\begin{align*} 
&\quad\;\int_0^{\infty}\Big[\frac{a^*}{2} (H_{\ep}(z))^2+H_{\ep}(z)\Big]\dd F(z)-(1+\theta)\int_0^{\infty}H_{\ep}(z)\dd\pw(z)\\
&\leq \ep\Big\{\int_0^{\infty}\Big[\frac{a^*}{2} (H_{1}(z))^2+H_{1}(z)\Big]\dd F(z)-(1+\theta)\int_0^{\infty}H_{1}(z)\dd\pw(z)\Big\}\\
&\quad\;+(1-\ep)\Big\{\int_0^{\infty}\Big[\frac{a^*}{2} (H^*(z))^2+ H^*(z)\Big]\dd F(z)-(1+\theta)\int_0^{\infty} H^*(z)\dd\pw(z)\Big\}\\
&\leq \int_0^{\infty}\Big[\frac{a^*}{2} (H^*(z))^2 + H^*(z)\Big]\dd F(z)-(1+\theta)\int_0^{\infty} H^*(z)\dd\pw(z)-c\ep,
\end{align*} 
and thus 
\begin{align} 
\label{eq:con}
\qquad\liminf_{\ep\to 0+} \frac{1}{\ep}&\bigg[\Big\{ \int_0^{\infty}\Big[\frac{a^*}{2} (H_{\ep}(z))^2+H_{\ep}(z)\Big]\dd F(z)-(1+\theta)\int_0^{\infty}H_{\ep}(z)\dd\pw(z)\Big\} \\
&-\Big\{\int_0^{\infty}\Big[\frac{a^*}{2} (H^*(z))^2 + H^*(z)\Big]\dd F(z)-(1+\theta)\int_0^{\infty} H^*(z)\dd\pw(z)\Big\}\bigg]\leq-c<0. \notag
\end{align} 
However, by \eqref{L2a}, \eqref{L2b}, and the 
definition of $H_\ep$, we see that the left hand side of the above inequality in \eqref{eq:con} is equal to 
\begin{align*} 
&\quad\;\liminf_{\ep\to 0+} \bigg[\int_0^{\infty}\big(H_1(z)- H^*(z)\big)\big[(a^* H^*(z)+1)\dd F(z)-(1+\theta)\dd\pw(z)\big]\\
&\qquad\qquad\qquad\qquad+\ep \int_0^{\infty} \frac{a^*}{2}(H_1(z)- H^*(z))^2 \dd F(z)\bigg]\\
&=\int_0^{\infty}\big(H_1(z) - H^*(z)\big)\big[(a^* H^*(z)+1)\dd F(z)-(1+\theta)\dd\pw(z)\big],
\end{align*}
which is nonnegative due to \eqref{eq:op_I}, and thus contradicts \eqref{eq:con}. 
This completes the proof. 
\end{proof}

\begin{proof}[\underline{Proof of Theorem \ref{thm:op2}}]
With $\Phi$ defined in \eqref{opldef}, the optimality condition in \eqref{eq:op_I} is equivalent to 
$\int_0^{\infty} \big(H(z)- H^*(z)\big)\dd\opl(z)\leq 0$.
Since $\opl(\infty)=0$ and $H(0)=H^*(0)=0$, applying Fubini's theorem to the previous inequality gives 
\begin{align*} 
\int_{0}^{\infty} \opl(z)\big(H'(z)- (H^*)'(z)\big)\dz\geq 0.
\end{align*} 
Because 
\begin{align*} 
\int_{0}^{\infty} \opl(z) (H^*)'(z) \dz&=-\int_{0}^{\infty} H^*(z) \opl'(z) \dz
\geq -(1+\theta)\int_{0}^{\infty}z\dd\pw(z)
>-\infty,
\end{align*} 
we get 
\begin{align*} 
\int_{0}^{\infty} \opl(z) H'(z) \dz\geq \int_{0}^{\infty} \opl(z) (H^*)'(z) \dz>-\infty, \quad \forall H \in \Ac. 
\end{align*} 
The above result implies that the linear functional 
\begin{align} \label{optimalcondition1-2}
H\mapsto \int_{0}^{\infty} \opl(z)H'(z)\dz, \quad H\in \Ac, 
\end{align} 
is minimized at $H^*$. Because $0\leq H'\leq 1$ for $H\in\Ac$, we conclude that 
\begin{align} \label{optimalcondition2}
\begin{cases}
(H^*)'(z)=1, &\;\text{if } \opl(z)<0;\\
(H^*)'(z)\in[0, 1], &\;\text{if } \opl(z)=0;\\
(H^*)'(z)=0, &\;\text{if } \opl(z)>0, 
\end{cases}
\end{align} 
for a.e. $z>0$.
By Lemma \ref{lem:in_eq}, the condition in \eqref{optimalcondition2} is equivalent to the OIDE in \eqref{vi001}.
\end{proof}

\begin{lemma}[Lemma 4.2, \citet{xu2021pareto}]
\label{lem:in_eq}
Let $a, b, c, d$ be real numbers and assume $b \le c$.
Then, 
$\min\{\max\{a-c, \; d\}, \;a-b\}=0$ if and only if $a=c$ when $d <0$, $a \in [b, c]$ when $d = 0$, and $a = b$ when $d > 0$. 
\end{lemma}

\begin{proof}[\underline{Proof of Corollary \ref{coro1}}]
In view of Theorem \ref{thm:veri}, it suffices to show $H^*(z)=z- I^*(z)=\min\{d, \; z\}$ 
is an optimal solution to the infimum problem in \eqref{opi1} with $a = a^*$. By Theorem \ref{thm:op2}, it remains to show the above $H^*$ satisfies the OIDE \eqref{vi001}, which is equivalent to proving $\opl(z; H^*)\geq 0$ for $z\in [d, \infty)$ and $\opl(z; H^*)\leq 0$ for $z\in [0, d)$. 

Using the definition of $\Phi$ in \eqref{opldef} and the above $H^*$, we compute
\begin{align*}
\opl(z; H^*)&=\int_z^{\infty}\big[(a^* \min\{d, \; x\}+1)\dd F(x)-(1+\theta)\dd\pw(x)\big]\\
&=-a^* \int_z^{d} (d-x)^+\dd F(x)+(a^* d+1)(1-F(z))-(1+\theta_0)\wg(1-F(z)). 
\end{align*}
To proceed, we consider two separate scenarios based on the value of $z$. 
\\[1ex]
\underline{Scenario 1}. When $z\in [d, \infty)$, we simplify to get $\opl(z; H^*) =(a^* d+1)(1-F(z))-(1+\theta_0)\wg(1-F(z))$. 
Hence, $\opl(z; H^*)\geq 0$ is equivalent to the second condition in \eqref{deductible}. 
\\[1ex]
\underline{Scenario 2}. When $z \in [0, d)$, 
\begin{align*} 
\opl(z; H^*) &=a^*\int_z^{\infty} (1-F(x)) (H^*(x))' \dx+(1-F(z))(a H^*(z)+1)-(1+\theta_0)\wg(1-F(z))\\
&=a^*\int_z^{d} (1-F(x))\dx+(1-F(z))(a^* z+1)-(1+\theta_0)\wg(1-F(z))\\
&=a\int_z^{d} (F(z)-F(x))\dx+(1-F(z))(a^*d+1)-(1+\theta_0)\wg(1-F(z))\\
&=-a^*\int_0^{d} (F(x)-F(z))^+\dx+(1-F(z))(a^*d+1)-(1+\theta_0)\wg(1-F(z)).
\end{align*} 
As such, $\Phi(z;H^*) \le 0$ if and only if the first condition in \eqref{deductible} holds. 

Finally, the result of $a^*$ follows from \eqref{eq:v_star} and 
\begin{align*} 
v(a^*)
&=\int_0^{\infty}\Big[\frac{a^*}{2}\min\{d, z\}^2+\min\{d, z\}\Big]\dd F(z)-(1+\theta)\int_0^{\infty}\min\{d, z\}\dd\pw(z)\\
&=\int_0^{\infty}\Big[\frac{a^*}{2}\min\{d, z\}^2+z-(z-d)^+\Big]\dd F(z)-(1+\theta)\int_0^{\infty}(z-(z-d)^+)\dd\pw(z).
\end{align*} 
\end{proof}

\small 
\setlength{\bibsep}{0pt plus 0.3ex}
\bibliographystyle{abbrvnat}
\bibliography{reference}

\begin{thebibliography}{27}
\providecommand{\natexlab}[1]{#1}
\providecommand{\url}[1]{\texttt{#1}}
\expandafter\ifx\csname urlstyle\endcsname\relax
  \providecommand{\doi}[1]{doi: #1}\else
  \providecommand{\doi}{doi: \begingroup \urlstyle{rm}\Url}\fi

\bibitem[Albrecher et~al.(2017)Albrecher, Beirlant, and
  Teugels]{albrecher2017reinsurance}
H.~Albrecher, J.~Beirlant, and J.~L. Teugels.
\newblock \emph{Reinsurance: Actuarial and Statistical Aspects}.
\newblock John Wiley \& Sons, 2017.

\bibitem[Arrow(1963)]{arrow1963uncertainty}
K.~J. Arrow.
\newblock Uncertainty and the welfare economics of medical care.
\newblock \emph{American Economic Review}, 53\penalty0 (5):\penalty0 941--973,
  1963.

\bibitem[Bayraktar and Young(2007)]{bayraktar2007minimizing}
E.~Bayraktar and V.~R. Young.
\newblock Minimizing the probability of lifetime ruin under borrowing
  constraints.
\newblock \emph{Insurance: Mathematics and Economics}, 41\penalty0
  (1):\penalty0 196--221, 2007.

\bibitem[Bayraktar and Zhang(2015)]{bayraktar2015stochastic}
E.~Bayraktar and Y.~Zhang.
\newblock Stochastic perron's method for the probability of lifetime ruin
  problem under transaction costs.
\newblock \emph{SIAM Journal on Control and Optimization}, 53\penalty0
  (1):\penalty0 91--113, 2015.

\bibitem[Bernard and Tian(2009)]{bernard2009optimal}
C.~Bernard and W.~Tian.
\newblock Optimal reinsurance arrangements under tail risk measures.
\newblock \emph{Journal of Risk and Insurance}, 76\penalty0 (3):\penalty0
  709--725, 2009.

\bibitem[Bernard et~al.(2015)Bernard, He, Yan, and Zhou]{bernard2015optimal}
C.~Bernard, X.~He, J.-A. Yan, and X.~Y. Zhou.
\newblock Optimal insurance design under rank-dependent expected utility.
\newblock \emph{Mathematical Finance}, 25\penalty0 (1):\penalty0 154--186,
  2015.

\bibitem[Boonen and Ghossoub(2019)]{boonen2019existence}
T.~J. Boonen and M.~Ghossoub.
\newblock On the existence of a representative reinsurer under heterogeneous
  beliefs.
\newblock \emph{Insurance: Mathematics and Economics}, 88:\penalty0 209--225,
  2019.

\bibitem[Boonen and Ghossoub(2023)]{boonen2023bowley}
T.~J. Boonen and M.~Ghossoub.
\newblock Bowley vs. pareto optima in reinsurance contracting.
\newblock \emph{European Journal of Operational Research}, 307\penalty0
  (1):\penalty0 382--391, 2023.

\bibitem[Boonen and Jiang(2022)]{boonen2022marginal}
T.~J. Boonen and W.~Jiang.
\newblock A marginal indemnity function approach to optimal reinsurance under
  the {Vajda} condition.
\newblock \emph{European Journal of Operational Research}, 303\penalty0
  (2):\penalty0 928--944, 2022.

\bibitem[Borch(1962)]{borch1962equilibrium}
K.~Borch.
\newblock Equilibrium in a reinsurance market.
\newblock \emph{Econometrica}, 30\penalty0 (3):\penalty0 424--444, 1962.

\bibitem[Cai and Chi(2020)]{cai2020optimal}
J.~Cai and Y.~Chi.
\newblock Optimal reinsurance designs based on risk measures: {A} review.
\newblock \emph{Statistical Theory and Related Fields}, 4\penalty0
  (1):\penalty0 1--13, 2020.

\bibitem[Carlier and Dana(2003)]{carlier2003pareto}
G.~Carlier and R.-A. Dana.
\newblock Pareto efficient insurance contracts when the insurer's cost function
  is discontinuous.
\newblock \emph{Economic Theory}, 21\penalty0 (4):\penalty0 871--893, 2003.

\bibitem[Chi and Zhuang(2020)]{chi2020optimal}
Y.~Chi and S.~C. Zhuang.
\newblock Optimal insurance with belief heterogeneity and incentive
  compatibility.
\newblock \emph{Insurance: Mathematics and Economics}, 92:\penalty0 104--114,
  2020.

\bibitem[Chi and Zhuang(2022)]{chi2022regret}
Y.~Chi and S.~C. Zhuang.
\newblock Regret-based optimal insurance design.
\newblock \emph{Insurance: Mathematics and Economics}, 102:\penalty0 22--41,
  2022.

\bibitem[Gajek and Zagrodny(2004)]{gajek2004reinsurance}
L.~Gajek and D.~Zagrodny.
\newblock Reinsurance arrangements maximizing insurer's survival probability.
\newblock \emph{Journal of Risk and Insurance}, 71\penalty0 (3):\penalty0
  421--435, 2004.

\bibitem[Grandell(1991)]{grandell1991aspects}
J.~Grandell.
\newblock \emph{Aspects of Risk Theory}.
\newblock Springer-Verlag New York, 1991.

\bibitem[Huberman et~al.(1983)Huberman, Mayers, and
  Smith~Jr]{huberman1983optimal}
G.~Huberman, D.~Mayers, and C.~W. Smith~Jr.
\newblock Optimal insurance policy indemnity schedules.
\newblock \emph{Bell Journal of Economics}, 14\penalty0 (2):\penalty0 415--426,
  1983.

\bibitem[Liang and Young(2020)]{liang2020minimizing}
X.~Liang and V.~R. Young.
\newblock Minimizing the discounted probability of exponential parisian ruin
  via reinsurance.
\newblock \emph{SIAM Journal on Control and Optimization}, 58\penalty0
  (2):\penalty0 937--964, 2020.

\bibitem[Promislow and Young(2005)]{promislow2005minimizing}
D.~S. Promislow and V.~R. Young.
\newblock Minimizing the probability of ruin when claims follow brownian motion
  with drift.
\newblock \emph{North American Actuarial Journal}, 9\penalty0 (3):\penalty0
  110--128, 2005.

\bibitem[Quiggin(1982)]{quiggin1982theory}
J.~Quiggin.
\newblock A theory of anticipated utility.
\newblock \emph{Journal of Economic Behavior \& Organization}, 3\penalty0
  (4):\penalty0 323--343, 1982.

\bibitem[Schmidli(2001)]{schmidli2001optimal}
H.~Schmidli.
\newblock Optimal proportional reinsurance policies in a dynamic setting.
\newblock \emph{Scandinavian Actuarial Journal}, 2001\penalty0 (1):\penalty0
  55--68, 2001.

\bibitem[Tversky and Kahneman(1992)]{tversky1992advances}
A.~Tversky and D.~Kahneman.
\newblock Advances in prospect theory: {Cumulative} representation of
  uncertainty.
\newblock \emph{Journal of Risk and Uncertainty}, 5:\penalty0 297--323, 1992.

\bibitem[Wang(1996)]{wang1996premium}
S.~S. Wang.
\newblock Premium calculation by transforming the layer premium density.
\newblock \emph{ASTIN Bulletin}, 26\penalty0 (1):\penalty0 71--92, 1996.

\bibitem[Wang et~al.(1997)Wang, Young, and Panjer]{wang1997axiomatic}
S.~S. Wang, V.~R. Young, and H.~H. Panjer.
\newblock Axiomatic characterization of insurance prices.
\newblock \emph{Insurance: Mathematics and Economics}, 21\penalty0
  (2):\penalty0 173--183, 1997.

\bibitem[Xu(2021)]{xu2021pareto}
Z.~Q. Xu.
\newblock Pareto optimal moral-hazard-free insurance contracts in behavioral
  finance framework.
\newblock \emph{Working paper}, available at
  \url{https://arxiv.org/abs/1803.02546}, 2021.

\bibitem[Xu(2022)]{xu2022moral}
Z.~Q. Xu.
\newblock Moral-hazard-free insurance: mean-variance premium principle and
  rank-dependent utility theory.
\newblock \emph{Scandinavian Actuarial Journal}, vol.2023 (2023), No.3, 269-289.

\bibitem[Xu et~al.(2019)Xu, Zhou, and Zhuang]{xu2019optimal}
Z.~Q. Xu, X.~Y. Zhou, and S.~C. Zhuang.
\newblock Optimal insurance under rank-dependent utility and incentive
  compatibility.
\newblock \emph{Mathematical Finance}, 29\penalty0 (2):\penalty0 659--692,
  2019.

\end{thebibliography}

\end{document}